\documentclass{article}

\usepackage{times}
\usepackage{graphicx} 
\usepackage{subfigure} 
\usepackage{cite}

\usepackage{algorithm}
\usepackage{algorithmic}
\usepackage{times}
\usepackage{graphicx} 
\usepackage{subfigure} 
\usepackage{amsmath,epsfig,cite} 
\usepackage{amsfonts}
\usepackage{amssymb}
\usepackage{amsthm}
\usepackage{mathrsfs}
\usepackage{subfigure}
\usepackage{amsmath}
\usepackage{paralist}
\usepackage{booktabs}
\usepackage{subfigure}
\usepackage{color}
\usepackage{url}
\DeclareGraphicsExtensions{.pdf,.jpeg,.png,.epbibdesk}
\usepackage{epstopdf}
\usepackage{natbib}
\usepackage[T1]{fontenc}
\usepackage{graphicx}
\usepackage{yhmath}
\usepackage{amsthm}
\usepackage{enumitem}
\newtheorem*{remark}{Remark}

\usepackage{amssymb}
\usepackage{amsfonts}
\usepackage{mathrsfs}
\usepackage{xspace}
\usepackage{bm}
\usepackage{upgreek}

\newcommand{\safemath}[2]{\newcommand{#1}{\ensuremath{#2}\xspace}}

\safemath{\bma}{\mathbf{a}}
\safemath{\bmb}{\mathbf{b}}
\safemath{\bmc}{\mathbf{c}}
\safemath{\bmd}{\mathbf{d}}
\safemath{\bme}{\mathbf{e}}
\safemath{\bmf}{\mathbf{f}}
\safemath{\bmg}{\mathbf{g}}
\safemath{\bmh}{\mathbf{h}}
\safemath{\bmi}{\mathbf{i}}
\safemath{\bmj}{\mathbf{j}}
\safemath{\bmk}{\mathbf{k}}
\safemath{\bml}{\mathbf{l}}
\safemath{\bmm}{\mathbf{m}}
\safemath{\bmn}{\mathbf{n}}
\safemath{\bmo}{\mathbf{o}}
\safemath{\bmp}{\mathbf{p}}
\safemath{\bmq}{\mathbf{q}}
\safemath{\bmr}{\mathbf{r}}
\safemath{\bms}{\mathbf{s}}
\safemath{\bmt}{\mathbf{t}}
\safemath{\bmu}{\mathbf{u}}
\safemath{\bmv}{\mathbf{v}}
\safemath{\bmw}{\mathbf{w}}
\safemath{\bmx}{\mathbf{x}}
\safemath{\bmy}{\mathbf{y}}
\safemath{\bmz}{\mathbf{z}}
\safemath{\bmzero}{\mathbf{0}}
\safemath{\bmone}{\mathbf{1}}

\bmdefine{\biad}{a}
\bmdefine{\bibd}{b}
\bmdefine{\bicd}{c}
\bmdefine{\bidd}{d}
\bmdefine{\bied}{e}
\bmdefine{\bifd}{f}
\bmdefine{\bigd}{g}
\bmdefine{\bihd}{h}
\bmdefine{\biid}{i}
\bmdefine{\bijd}{j}
\bmdefine{\bikd}{k}
\bmdefine{\bild}{l}
\bmdefine{\bimd}{m}
\bmdefine{\bind}{n}
\bmdefine{\biod}{o}
\bmdefine{\bipd}{p}
\bmdefine{\biqd}{q}
\bmdefine{\bird}{r}
\bmdefine{\bisd}{s}
\bmdefine{\bitd}{t}
\bmdefine{\biud}{u}
\bmdefine{\bivd}{v}
\bmdefine{\biwd}{w}
\bmdefine{\bixd}{x}
\bmdefine{\biyd}{y}
\bmdefine{\bizd}{z}

\bmdefine{\bixid}{\xi}
\bmdefine{\bilambdad}{\lambda}
\bmdefine{\bimud}{\mu}
\bmdefine{\bithetad}{\theta}
\bmdefine{\biphid}{\phi}
\bmdefine{\bideltad}{\delta}

\safemath{\bmia}{\biad}
\safemath{\bmib}{\bibd}
\safemath{\bmic}{\bicd}
\safemath{\bmid}{\bidd}
\safemath{\bmie}{\bied}
\safemath{\bmif}{\bifd}
\safemath{\bmig}{\bigd}
\safemath{\bmih}{\bihd}
\safemath{\bmii}{\biid}
\safemath{\bmij}{\bijd}
\safemath{\bmik}{\bikd}
\safemath{\bmil}{\bild}
\safemath{\bmim}{\bimd}
\safemath{\bmin}{\bind}
\safemath{\bmio}{\biod}
\safemath{\bmip}{\bipd}
\safemath{\bmiq}{\biqd}
\safemath{\bmir}{\bird}
\safemath{\bmis}{\bisd}
\safemath{\bmit}{\bitd}
\safemath{\bmiu}{\biud}
\safemath{\bmiv}{\bivd}
\safemath{\bmiw}{\biwd}
\safemath{\bmix}{\bixd}
\safemath{\bmiy}{\biyd}
\safemath{\bmiz}{\bizd}

\safemath{\bmxi}{\bixid}
\safemath{\bmlambda}{\bilambdad}
\safemath{\bmmu}{\bimud}
\safemath{\bmtheta}{\bithetad}
\safemath{\bmphi}{\biphid}
\safemath{\bmdelta}{\bideltad}

\safemath{\bA}{\mathbf{A}}
\safemath{\bB}{\mathbf{B}}
\safemath{\bC}{\mathbf{C}}
\safemath{\bD}{\mathbf{D}}
\safemath{\bE}{\mathbf{E}}
\safemath{\bF}{\mathbf{F}}
\safemath{\bG}{\mathbf{G}}
\safemath{\bH}{\mathbf{H}}
\safemath{\bI}{\mathbf{I}}
\safemath{\bJ}{\mathbf{J}}
\safemath{\bK}{\mathbf{K}}
\safemath{\bL}{\mathbf{L}}
\safemath{\bM}{\mathbf{M}}
\safemath{\bN}{\mathbf{N}}
\safemath{\bO}{\mathbf{O}}
\safemath{\bP}{\mathbf{P}}
\safemath{\bQ}{\mathbf{Q}}
\safemath{\bR}{\mathbf{R}}
\safemath{\bS}{\mathbf{S}}
\safemath{\bT}{\mathbf{T}}
\safemath{\bU}{\mathbf{U}}
\safemath{\bV}{\mathbf{V}}
\safemath{\bW}{\mathbf{W}}
\safemath{\bX}{\mathbf{X}}
\safemath{\bY}{\mathbf{Y}}
\safemath{\bZ}{\mathbf{Z}}

\safemath{\bZero}{\mathbf{0}}
\safemath{\bOne}{\mathbf{1}}
\safemath{\bDelta}{\mathbf{\Delta}}
\safemath{\bLambda}{\mathbf{\UpLambda}}
\safemath{\bPhi}{\mathbf{\Upphi}}
\safemath{\bSigma}{\mathbf{\Upsigma}}
\safemath{\bOmega}{\mathbf{\Upomega}}
\safemath{\bTheta}{\mathbf{\Uptheta}}

\bmdefine{\biAd}{A}
\bmdefine{\biBd}{B}
\bmdefine{\biCd}{C}
\bmdefine{\biDd}{D}
\bmdefine{\biEd}{E}
\bmdefine{\biFd}{F}
\bmdefine{\biGd}{G}
\bmdefine{\biHd}{H}
\bmdefine{\biId}{I}
\bmdefine{\biJd}{J}
\bmdefine{\biKd}{K}
\bmdefine{\biLd}{L}
\bmdefine{\biMd}{M}
\bmdefine{\biNd}{N}
\bmdefine{\biOd}{O}
\bmdefine{\biPd}{P}
\bmdefine{\biQd}{Q}
\bmdefine{\biRd}{R}
\bmdefine{\biSd}{S}
\bmdefine{\biTd}{T}
\bmdefine{\biUd}{U}
\bmdefine{\biVd}{V}
\bmdefine{\biWd}{W}
\bmdefine{\biXd}{X}
\bmdefine{\biYd}{Y}
\bmdefine{\biZd}{Z}

\bmdefine{\biDelta}{\Delta}
\bmdefine{\biLambda}{\Lambda}
\bmdefine{\biPhi}{\Phi}
\bmdefine{\biSigma}{\Sigma}
\bmdefine{\biOmega}{\Omega}
\bmdefine{\biTheta}{\Theta}

\safemath{\bimA}{\biAd}
\safemath{\bimB}{\biBd}
\safemath{\bimC}{\biCd}
\safemath{\bimD}{\biDd}
\safemath{\bimE}{\biEd}
\safemath{\bimF}{\biFd}
\safemath{\bimG}{\biGd}
\safemath{\bimH}{\biHd}
\safemath{\bimI}{\biId}
\safemath{\bimJ}{\biJd}
\safemath{\bimK}{\biKd}
\safemath{\bimL}{\biLd}
\safemath{\bimM}{\biMd}
\safemath{\bimN}{\biNd}
\safemath{\bimO}{\biOd}
\safemath{\bimP}{\biPd}
\safemath{\bimQ}{\biQd}
\safemath{\bimR}{\biRd}
\safemath{\bimS}{\biSd}
\safemath{\bimT}{\biTd}
\safemath{\bimU}{\biUd}
\safemath{\bimV}{\biVd}
\safemath{\bimW}{\biWd}
\safemath{\bimX}{\biXd}
\safemath{\bimY}{\biYd}
\safemath{\bimZ}{\biZd}

\safemath{\bimDelta}{\biDelta}
\safemath{\bimLambda}{\biLambda}
\safemath{\bimPhi}{\biPhi}
\safemath{\bimSigma}{\biSigma}
\safemath{\bimOmega}{\biOmega}
\safemath{\bimTheta}{\biTheta}

\safemath{\setA}{\mathcal{A}}
\safemath{\setB}{\mathcal{B}}
\safemath{\setC}{\mathcal{C}}
\safemath{\setD}{\mathcal{D}}
\safemath{\setE}{\mathcal{E}}
\safemath{\setF}{\mathcal{F}}
\safemath{\setG}{\mathcal{G}}
\safemath{\setH}{\mathcal{H}}
\safemath{\setI}{\mathcal{I}}
\safemath{\setJ}{\mathcal{J}}
\safemath{\setK}{\mathcal{K}}
\safemath{\setL}{\mathcal{L}}
\safemath{\setM}{\mathcal{M}}
\safemath{\setN}{\mathcal{N}}
\safemath{\setO}{\mathcal{O}}
\safemath{\setP}{\mathcal{P}}
\safemath{\setQ}{\mathcal{Q}}
\safemath{\setR}{\mathcal{R}}
\safemath{\setS}{\mathcal{S}}
\safemath{\setT}{\mathcal{T}}
\safemath{\setU}{\mathcal{U}}
\safemath{\setV}{\mathcal{V}}
\safemath{\setW}{\mathcal{W}}
\safemath{\setX}{\mathcal{X}}
\safemath{\setY}{\mathcal{Y}}
\safemath{\setZ}{\mathcal{Z}}
\safemath{\emptySet}{\varnothing}

\safemath{\colA}{\mathscr{A}}
\safemath{\colB}{\mathscr{B}}
\safemath{\colC}{\mathscr{C}}
\safemath{\colD}{\mathscr{D}}
\safemath{\colE}{\mathscr{E}}
\safemath{\colF}{\mathscr{F}}
\safemath{\colG}{\mathscr{G}}
\safemath{\colH}{\mathscr{H}}
\safemath{\colI}{\mathscr{I}}
\safemath{\colJ}{\mathscr{J}}
\safemath{\colK}{\mathscr{K}}
\safemath{\colL}{\mathscr{L}}
\safemath{\colM}{\mathscr{M}}
\safemath{\colN}{\mathscr{N}}
\safemath{\colO}{\mathscr{O}}
\safemath{\colP}{\mathscr{P}}
\safemath{\colQ}{\mathscr{Q}}
\safemath{\colR}{\mathscr{R}}
\safemath{\colS}{\mathscr{S}}
\safemath{\colT}{\mathscr{T}}
\safemath{\colU}{\mathscr{U}}
\safemath{\colV}{\mathscr{V}}
\safemath{\colW}{\mathscr{W}}
\safemath{\colX}{\mathscr{X}}
\safemath{\colY}{\mathscr{Y}}
\safemath{\colZ}{\mathscr{Z}}

\safemath{\opA}{\mathbb{A}}
\safemath{\opB}{\mathbb{B}}
\safemath{\opC}{\mathbb{C}}
\safemath{\opD}{\mathbb{D}}
\safemath{\opE}{\mathbb{E}}
\safemath{\opF}{\mathbb{F}}
\safemath{\opG}{\mathbb{G}}
\safemath{\opH}{\mathbb{H}}
\safemath{\opI}{\mathbb{I}}
\safemath{\opJ}{\mathbb{J}}
\safemath{\opK}{\mathbb{K}}
\safemath{\opL}{\mathbb{L}}
\safemath{\opM}{\mathbb{M}}
\safemath{\opN}{\mathbb{N}}
\safemath{\opO}{\mathbb{O}}
\safemath{\opP}{\mathbb{P}}
\safemath{\opQ}{\mathbb{Q}}
\safemath{\opR}{\mathbb{R}}
\safemath{\opS}{\mathbb{S}}
\safemath{\opT}{\mathbb{T}}
\safemath{\opU}{\mathbb{U}}
\safemath{\opV}{\mathbb{V}}
\safemath{\opW}{\mathbb{W}}
\safemath{\opX}{\mathbb{X}}
\safemath{\opY}{\mathbb{Y}}
\safemath{\opZ}{\mathbb{Z}}
\safemath{\opZero}{\mathbb{O}}
\safemath{\identityop}{\opI}

\safemath{\veca}{\bma}
\safemath{\vecb}{\bmb}
\safemath{\vecc}{\bmc}
\safemath{\vecd}{\bmd}
\safemath{\vece}{\bme}
\safemath{\vecf}{\bmf}
\safemath{\vecg}{\bmg}
\safemath{\vech}{\bmh}
\safemath{\veci}{\bmi}
\safemath{\vecj}{\bmj}
\safemath{\veck}{\bmk}
\safemath{\vecl}{\bml}
\safemath{\vecm}{\bmm}
\safemath{\vecn}{\bmn}
\safemath{\veco}{\bmo}
\safemath{\vecp}{\bmp}
\safemath{\vecq}{\bmq}
\safemath{\vecr}{\bmr}
\safemath{\vecs}{\bms}
\safemath{\vect}{\bmt}
\safemath{\vecu}{\bmu}
\safemath{\vecv}{\bmv}
\safemath{\vecw}{\bmw}
\safemath{\vecx}{\bmx}
\safemath{\vecy}{\bmy}
\safemath{\vecz}{\bmz}

\safemath{\veczero}{\bmzero}
\safemath{\vecone}{\bmone}
\safemath{\vecxi}{\bmxi}
\safemath{\veclambda}{\bmlambda}
\safemath{\vecmu}{\bmmu}
\safemath{\vectheta}{\bmtheta}
\safemath{\vecphi}{\bmphi}
\safemath{\vecdelta}{\bmdelta}

\safemath{\matA}{\bA}
\safemath{\matB}{\bB}
\safemath{\matC}{\bC}
\safemath{\matD}{\bD}
\safemath{\matE}{\bE}
\safemath{\matF}{\bF}
\safemath{\matG}{\bG}
\safemath{\matH}{\bH}
\safemath{\matI}{\bI}
\safemath{\matJ}{\bJ}
\safemath{\matK}{\bK}
\safemath{\matL}{\bL}
\safemath{\matM}{\bM}
\safemath{\matN}{\bN}
\safemath{\matO}{\bO}
\safemath{\matP}{\bP}
\safemath{\matQ}{\bQ}
\safemath{\matR}{\bR}
\safemath{\matS}{\bS}
\safemath{\matT}{\bT}
\safemath{\matU}{\bU}
\safemath{\matV}{\bV}
\safemath{\matW}{\bW}
\safemath{\matX}{\bX}
\safemath{\matY}{\bY}
\safemath{\matZ}{\bZ}
\safemath{\matzero}{\bmzero}

\safemath{\matDelta}{\bDelta}
\safemath{\matLambda}{\bLambda}
\safemath{\matPhi}{\bPhi}
\safemath{\matSigma}{\bSigma}
\safemath{\matOmega}{\bOmega}
\safemath{\matTheta}{\bTheta}

\safemath{\matidentity}{\matI}
\safemath{\matone}{\matO}

\safemath{\rnda}{A}
\safemath{\rndb}{B}
\safemath{\rndc}{C}
\safemath{\rndd}{D}
\safemath{\rnde}{E}
\safemath{\rndf}{F}
\safemath{\rndg}{G}
\safemath{\rndh}{H}
\safemath{\rndi}{I}
\safemath{\rndj}{J}
\safemath{\rndk}{K}
\safemath{\rndl}{L}
\safemath{\rndm}{M}
\safemath{\rndn}{N}
\safemath{\rndo}{O}
\safemath{\rndp}{P}
\safemath{\rndq}{Q}
\safemath{\rndr}{R}
\safemath{\rnds}{S}
\safemath{\rndt}{T}
\safemath{\rndu}{U}
\safemath{\rndv}{V}
\safemath{\rndw}{W}
\safemath{\rndx}{X}
\safemath{\rndy}{Y}
\safemath{\rndz}{Z}

\safemath{\rveca}{\bimA}
\safemath{\rvecb}{\bimB}
\safemath{\rvecc}{\bimC}
\safemath{\rvecd}{\bimD}
\safemath{\rvece}{\bimE}
\safemath{\rvecf}{\bimF}
\safemath{\rvecg}{\bimG}
\safemath{\rvech}{\bimH}
\safemath{\rveci}{\bimI}
\safemath{\rvecj}{\bimJ}
\safemath{\rveck}{\bimK}
\safemath{\rvecl}{\bimL}
\safemath{\rvecm}{\bimM}
\safemath{\rvecn}{\bimN}
\safemath{\rveco}{\bomO}
\safemath{\rvecp}{\bimP}
\safemath{\rvecq}{\bimQ}
\safemath{\rvecr}{\bimR}
\safemath{\rvecs}{\bimS}
\safemath{\rvect}{\bimT}
\safemath{\rvecu}{\bimU}
\safemath{\rvecv}{\bimV}
\safemath{\rvecw}{\bimW}
\safemath{\rvecx}{\bimX}
\safemath{\rvecy}{\bimY}
\safemath{\rvecz}{\bimZ}

\safemath{\rvecxi}{\bmxi}
\safemath{\rveclambda}{\bmlambda}
\safemath{\rvecmu}{\bmmu}
\safemath{\rvectheta}{\bmtheta}
\safemath{\rvecphi}{\bmphi}

\safemath{\rmatA}{\bimA}
\safemath{\rmatB}{\bimB}
\safemath{\rmatC}{\bimC}
\safemath{\rmatD}{\bimD}
\safemath{\rmatE}{\bimE}
\safemath{\rmatF}{\bimF}
\safemath{\rmatG}{\bimG}
\safemath{\rmatH}{\bimH}
\safemath{\rmatI}{\bimI}
\safemath{\rmatJ}{\bimJ}
\safemath{\rmatK}{\bimK}
\safemath{\rmatL}{\bimL}
\safemath{\rmatM}{\bimM}
\safemath{\rmatN}{\bimN}
\safemath{\rmatO}{\bimO}
\safemath{\rmatP}{\bimP}
\safemath{\rmatQ}{\bimQ}
\safemath{\rmatR}{\bimR}
\safemath{\rmatS}{\bimS}
\safemath{\rmatT}{\bimT}
\safemath{\rmatU}{\bimU}
\safemath{\rmatV}{\bimV}
\safemath{\rmatW}{\bimW}
\safemath{\rmatX}{\bimX}
\safemath{\rmatY}{\bimY}
\safemath{\rmatZ}{\bimZ}

\safemath{\rmatDelta}{\bimDelta}
\safemath{\rmatLambda}{\bimLambda}
\safemath{\rmatPhi}{\bimPhi}
\safemath{\rmatSigma}{\bimSigma}
\safemath{\rmatOmega}{\bimOmega}
\safemath{\rmatTheta}{\bimTheta}

\usepackage{amssymb}
\usepackage{amsfonts}
\usepackage{mathrsfs}
\usepackage{xspace}
\usepackage{bm}
\usepackage{fancyref}
\usepackage{textcomp}

\usepackage{multirow}
\usepackage{stmaryrd}

\newenvironment{textbmatrix}{	\setlength{\arraycolsep}{2.5pt}%
								\big[\begin{matrix}}{\end{matrix}\big]%
								\raisebox{0.08ex}{\vphantom{M}}}

\def\be{\begin{equation}}
\def\ee{\end{equation}}
\def\een{\nonumber \end{equation}}
\def\mat{\begin{bmatrix}}
\def\emat{\end{bmatrix}}
\def\btm{\begin{textbmatrix}}
\def\etm{\end{textbmatrix}}

\def\ba#1\ea{\begin{align}#1\end{align}}
\def\bas#1\eas{\begin{align*}#1\end{align*}}
\def\bs#1\es{\begin{split}#1\end{split}} 
\def\bg#1\eg{\begin{gather}#1\end{gather}}
\def\bml#1\eml{\begin{multline}#1\end{multline}}
\def\bi#1\ei{\begin{itemize}#1\end{itemize}}

\newcommand{\lefto}{\mathopen{}\left}

\newcommand{\abs}[1]{\lefto\lvert#1\right\rvert}		

\newcommand{\vecnorm}[1]{\lefto\lVert#1\right\rVert}		

\safemath{\dirac}{\delta}					
\safemath{\krond}{\dirac}					

\safemath{\upto}{\uparrow}
\safemath{\downto}{\downarrow}
\safemath{\iu}{j}							
\safemath{\ev}{\lambda}						
\safemath{\hilseqspace}{l^{2}}				
\newcommand{\banachfunspace}[1]{\setL^{#1}}	
\safemath{\hilfunspace}{\banachfunspace{2}}	

\safemath{\SNR}{\text{\sc snr}} 				
\safemath{\No}{N_0}							
\safemath{\Es}{E_s}							
\safemath{\Eb}{E_b}							
\safemath{\EbNo}{\frac{\Eb}{\No}}
\safemath{\EsNo}{\frac{\Es}{\No}}

\DeclareMathOperator{\CHop}{\ensuremath{\opH}} 
\safemath{\tvir}{\rndh_{\CHop}}				
\safemath{\tvtf}{\rndl_{\CHop}}				
\safemath{\spf}{\rnds_{\CHop}}				
\safemath{\bff}{H_{\CHop}}					

\safemath{\ircf}{r_{h}}						
\safemath{\tftvcf}{r_{s}}					
\safemath{\tfcf}{r_{l}}						
\safemath{\bfcf}{r_{H}}						

\safemath{\tcorr}{c_h}						
\safemath{\scf}{c_{s}}						
\safemath{\tfcorr}{c_{l}}					
\safemath{\fcorr}{c_{H}}						

\safemath{\mi}{I}							
\safemath{\capacity}{C}						

\safemath{\normal}{\mathcal{N}}			
\safemath{\jpg}{\mathcal{CN}}			
\safemath{\mchain}{\leftrightarrow}		

\safemath{\dB}{\,\mathrm{dB}}
\safemath{\dBm}{\,\mathrm{dBm}}
\safemath{\Hz}{\,\mathrm{Hz}}
\safemath{\kHz}{\,\mathrm{kHz}}
\safemath{\MHz}{\,\mathrm{MHz}}
\safemath{\GHz}{\,\mathrm{GHz}}
\safemath{\s}{\,\mathrm{s}}
\safemath{\ms}{\,\mathrm{ms}}
\safemath{\mus}{\,\mathrm{\text{\textmu}s}}
\safemath{\ns}{\,\mathrm{ns}}
\safemath{\ps}{\,\mathrm{ps}}
\safemath{\meter}{\,\mathrm{m}}
\safemath{\mm}{\,\mathrm{mm}}
\safemath{\cm}{\,\mathrm{cm}}
\safemath{\m}{\,\mathrm{m}}
\safemath{\W}{\,\mathrm{W}}
\safemath{\mW}{\, \mathrm{mW}}
\safemath{\J}{\,\mathrm{J}}
\safemath{\K}{\,\mathrm{K}}
\safemath{\bit}{\,\mathrm{bit}}
\safemath{\nat}{\,\mathrm{nat}}

\safemath{\define}{\triangleq}			

\safemath{\equivalent}{\sim}
\safemath{\distas}{\sim}					
\safemath{\sdiff}{\Delta}				

\safemath{\reals}{\mathbb{R}}
\safemath{\positivereals}{\reals_{+}}
\safemath{\integers}{\mathbb{Z}}
\safemath{\posint}{\integers_{+}}
\safemath{\naturals}{\mathbb{N}}
\safemath{\posnaturals}{\naturals_{+}}
\safemath{\complexset}{\mathbb{C}}
\safemath{\rationals}{\mathbb{Q}}

\newcommand*{\fancyrefapplabelprefix}{app}		
\newcommand*{\fancyrefthmlabelprefix}{thm}		
\newcommand*{\fancyreflemlabelprefix}{lem}		
\newcommand*{\fancyrefcorlabelprefix}{cor}		
\newcommand*{\fancyrefdeflabelprefix}{def}		
\newcommand*{\fancyrefalglabelprefix}{alg}		
\newcommand*{\fancyrefproplabelprefix}{prop}		
\newcommand*{\fancyrefexmpllabelprefix}{exmpl}
\newcommand*{\fancyreftbllabelprefix}{tbl}
\frefformat{vario}{\fancyrefseclabelprefix}{Sec.~#1}
\frefformat{vario}{\fancyrefthmlabelprefix}{Thm.~#1}
\frefformat{vario}{\fancyreflemlabelprefix}{Lem.~#1}
\frefformat{vario}{\fancyrefcorlabelprefix}{Cor.~#1}
\frefformat{vario}{\fancyrefdeflabelprefix}{Def.~#1}
\frefformat{vario}{\fancyrefalglabelprefix}{Alg.~#1}
\frefformat{vario}{\fancyreffiglabelprefix}{Fig.~#1}
\frefformat{vario}{\fancyrefapplabelprefix}{App.~#1} 
\frefformat{vario}{\fancyrefeqlabelprefix}{(#1)}
\frefformat{vario}{\fancyrefproplabelprefix}{Prop.~#1}
\frefformat{vario}{\fancyrefexmpllabelprefix}{Ex.~#1}
\frefformat{vario}{\fancyreftbllabelprefix}{Tbl.~#1}

 \newtheorem{thm}{Theorem}

 \newtheorem{lem}[thm]{Lemma}
 
\safemath{\dictab}{[\,\dicta\,\,\dictb\,]}

\safemath{\ysig}{\bmy}
\safemath{\ysighat}{\hat{\ysig}}
\safemath{\ysigdim}{M}
\safemath{\xsig}{\bmx}
\safemath{\xsigdim}{N}
\safemath{\nx}{n_x}
\safemath{\zsig}{\bmz}
\safemath{\zsigdim}{\ysigdim}
\safemath{\rsig}{\bmr}
\safemath{\Adict}{\bA}
\safemath{\Adicttilde}{\widetilde{\Adict}}
\safemath{\Adictdim}{\outputdim\times\xsigdim}
\safemath{\avec}{\bma}
\safemath{\avectilde}{\tilde{\avec}}
\safemath{\Bdict}{\bB}
\safemath{\Bdicttilde}{\widetilde{\Bdict}}
\safemath{\Cdict}{\bC}
\safemath{\cvec}{\bmc}
\safemath{\Ddict}{\bD}
\safemath{\Ddictdim}{\ysigdim\times\xsigdim}
\safemath{\dvec}{\bmd}
\safemath{\Ddicttilde}{\widetilde{\bD}}
\safemath{\Bonb}{\bB}
\safemath{\bvec}{\bmb}
\safemath{\Bonbdim}{\ysigdim\times\ysigdim}
\safemath{\noise}{\bmn}
\safemath{\noisedim}{\ysigim}
\safemath{\err}{\bme}
\safemath{\errdim}{\ysigdim}
\safemath{\errset}{\setE}
\safemath{\nerr}{n_e}
\safemath{\delop}{\bP_\errset}
\safemath{\delopc}{\bP_{{\errset}^c}}

\safemath{\cplxi}{\imath}
\safemath{\cplxj}{\jmath}

\safemath{\dict}{\matD}
\safemath{\inputdim}{N}		
\safemath{\outputdim}{M}		
\safemath{\sparsity}{S}	
\safemath{\inputdimA}{{N_a}}	
\safemath{\inputdimB}{{N_b}}	
\safemath{\elemA}{{n_a}}	
\safemath{\elemB}{{n_b}}	
\safemath{\resA}{\matR_a}	
\safemath{\resB}{\matR_b}	
\safemath{\subD}{\matS} 
\safemath{\subA}{\matS_a} 
\safemath{\subB}{\matS_b} 
\safemath{\dicta}{\matA} 	
\safemath{\dictb}{\matB} 	
\safemath{\hollowS}{H}
\safemath{\hollowA}{H_a}
\safemath{\hollowB}{H_b}
\safemath{\cross}{Z}
\safemath{\coh}{\mu_d}			
\safemath{\coha}{\mu_a}			
\safemath{\cohb}{\mu_b}			
\safemath{\mubs}{\nu}	
\safemath{\cohm}{\mu_m} 
\safemath{\dictset}{\setD}	
\safemath{\dictsetp}{\dictset(\coh,\coha,\cohb)}	
\safemath{\dictsetgen}{\dictset_\text{gen}}
\safemath{\dictsetgenp}{\dictsetgen(\coh)}
\safemath{\dictsetonb}{\dictset_\text{onb}}
\safemath{\dictsetonbp}{\dictsetonb(\coh)}

\safemath{\leftside}{U}
\safemath{\rightsideA}{R_a}
\safemath{\rightsideB}{R_b}

\safemath{\indexS}{\setI_S} 

\safemath{\na}{n_a}			
\safemath{\nb}{n_b}			
\safemath{\coeffa}{p_i}	
\safemath{\coeffb}{q_j}	
\safemath{\seta}{\setP}		
\safemath{\setb}{\setQ}     
\safemath{\setw}{\setW}	
\safemath{\setz}{\setZ}	
\safemath{\cola}{\veca}		
\safemath{\colb}{\vecb}		
\safemath{\cold}{\vecd}		
\safemath{\inputvec}{\vecx} 	
\safemath{\error}{\vece}	
\safemath{\noiseout}{\vecz} 	
\safemath{\inputvecel}{x}
\safemath{\inputveca}{\vecx_a}
\safemath{\inputvecb}{\vecx_b}
\safemath{\outputvec}{\vecy}	
\safemath{\lambdamin}{\lambda_{\mathrm{min}}}

\newcommand{\normtwo}[1]{\vecnorm{#1}_2}

\safemath{\elltwo}{\ell_2}
\safemath{\ellone}{\ell_1}
\safemath{\ellzero}{\ell_0}
\safemath{\ellinf}{\ell_\infty}
\safemath{\licard}{Z(\coh,\coha,\cohb)}
\safemath{\xsol}{\hat{x}}
\safemath{\xbord}{x_b}		
\safemath{\xstat}{x_s}		
\safemath{\xstatLone}{\tilde{x}_s}
\safemath{\order}{\mathcal{O}} 
\safemath{\scales}{\Theta} 
\safemath{\ones}{\mathbf{1}} 
\safemath{\zeroes}{\mathbf{0}} 
\safemath{\thlone}{\kappa(\coh,\cohb)} 
\safemath{\constoneA}{\delta} 
\safemath{\constoneB}{\epsilon} 
\safemath{\nlarge}{L}				   
\safemath{\sumlarge}{S_\nlarge}
\safemath{\maxlarger}{P_\nlarge}	   
\safemath{\Pzero}{\textrm{P0}}	
\safemath{\Pone}{\textrm{P1}}
\safemath{\vecfir}{\vecw}			 
\safemath{\vecsec}{\vecz}
\safemath{\elvecfir}{w}              
\safemath{\elvecsec}{z}				 
\safemath{\nlargefir}{n}
\safemath{\normout}{\gamma}
\safemath{\auxfun}{h}
\safemath{\supp}{\textrm{supp}}

\safemath{\indexa}{\ell}
\safemath{\indexb}{r}
\safemath{\indexc}{i}
\safemath{\indexd}{j}

\safemath{\project}{P}

\usepackage[accepted]{icml2016}

\icmltitlerunning{Near-Isometric Binary Hashing for Large-scale Datasets}

\begin{document} 

\twocolumn[
\icmltitle{Near-Isometric Binary Hashing for Large-scale Datasets}

\icmlauthor{Amirali Aghazadeh}{amirali@rice.edu}
\icmladdress{Department of Electrical and Computer Engineering, Rice University, Houston, TX, USA}
\icmlauthor{Andrew Lan}{mr.lan@sparfa.com}
\icmladdress{Department of Electrical and Computer Engineering, Rice University, Houston, TX, USA}
\icmlauthor{Anshumali Shrivastava}{anshumali@rice.edu}
\icmladdress{Department of Computer Science, Rice University, Houston, TX, USA}
\icmlauthor{Richard Baraniuk}{richb@rice.edu}
\icmladdress{Department of Electrical and Computer Engineering, Rice University, Houston, TX, USA}


\vskip 0.3in
]

\begin{abstract} 

We develop a scalable algorithm to learn binary hash codes for indexing large-scale datasets.
Near-isometric binary hashing (NIBH) is a data-dependent hashing scheme that quantizes the output of a learned low-dimensional embedding to obtain a binary hash code.
In contrast to conventional hashing schemes, which typically rely on an $\ell_2$-norm (i.e., average distortion) minimization, NIBH is based on a $\ell_{\infty}$-norm (i.e., worst-case distortion) minimization that provides several benefits, including superior distance, ranking, and near-neighbor preservation performance.
We develop a practical and efficient algorithm for NIBH based on column generation that scales well to large datasets.
A range of experimental evaluations demonstrate the superiority of NIBH over ten state-of-the-art binary hashing schemes.

\end{abstract} 

\section{Introduction} \label{sec:intro}

\sloppy 

Hashing, one of the primitive operations in large-scale systems, seeks a low-dimensional binary embedding of a high-dimensional data set.
Such a binary embedding can increase the computational efficiency of a variety of tasks, including searching, learning, near-neighbor retrieval, etc. 
One of the fundamental challenges in machine learning is the development of efficient hashing algorithms that embed data points into compact binary codes while preserving the geometry of the original dataset.

In this paper, we are interested in learning \emph{near-isometric} binary embeddings, i.e., hash functions that preserve the distances between data points up to a given distortion in the Hamming space.
More rigorously, let $\mathcal{Z}$ and $\mathcal{Y}$ denote metric spaces with metrics $d_\mathcal{Z}$ and $d_\mathcal{Y}$, respectively. 
An embedding $f: \mathcal{Z} \rightarrow \mathcal{Y}$ is called near-isometric~\citep[Def.~1.1]{yanivplan} if, for \emph{every} pair of data points $\vecz_i, \vecz_j \in \mathcal{Z}$, we have 
\begin{align*}
 d_\mathcal{Z}(\vecz_i,\vecz_j) - \gamma \! \leq \! d_\mathcal{Y}(f(\vecz_i),f(\vecz_j)) \! \leq \! d_\mathcal{Z}(\vecz_i,\vecz_j) + \gamma,
\end{align*}
where $\gamma$ is called the isometry constant.
In words, $f$ is near-isometric if and only if the entries of the \emph{pairwise-distortion vector} containing the distance distortion between every pair of data points 
$| d_\mathcal{Z}(\vecz_i,\vecz_j)-d_\mathcal{Y}(f(\vecz_i),f(\vecz_j))|, \forall i>j$ do not exceed the isometry constant $\gamma$.
A near-isometric embedding is approximately \emph{distance-preserving} in that the distance between any pairs of data points in the embedded space $\mathcal{Y}$ is approximately equal to their distance in the ambient space $\mathcal{Z}$~\cite{weinberger2006unsupervised,shaw2007minimum,numax}.

The simplest and most popular binary hashing scheme, {\em random projection}, simply projects the data into a lower-dimensional (lower-bit) random subspace and then quantizes to binary values.
Random projections are known to be near-isometric with high probability, due to the celebrated Johnson-Lindenstrauss (JL) lemma~\cite{lsh,andoni2014beyond,yanivplan}. 
Algorithms based on the JL lemma belong to the family of probabilistic dimensionality reduction techniques; a notable example is locality sensitive hashing (LSH) \cite{lsh,andoni2014beyond}.
Unfortunately, theoretical results on LSH state that the number of bits required to guarantee an isometric embedding can be as large as the number of data points~\cite{lsh,yanivplan}. 
Even in practice, LSH's requirement on the number of bits is impractically high for indexing many real-world, large-scale datasets \cite{lv2007multi}.

Consequently, several \emph{data-dependent} binary hashing algorithms have been developed that leverage the structure of the data to learn compact binary codes.
These methods enable a significant reduction in the number of bits required to index large-scale datasets compared to LSH; see~\cite{wang2014hashing} for a survey.
However, learning compact binary codes that preserve the local geometry of the data remains challenging.

These data-dependent hashing algorithms focus on the choice of the distortion measure. 
Typically, after finding the appropriate distortion measure, the hash functions are learned by minimizing the \emph{average} distortion, i.e., the $\ell_2$-norm of the pairwise-distortion vector, which sums the distortion among all pairwise distances with equal weights.
\emph{Binary reconstructive embedding} (BRE) \cite{kulis2009learning}, for example, uses an optimization algorithm to directly minimize the average distortion in the embedded sapce. 
\emph{Spectral Hashing} (SH) \cite{weiss2009spectral}, \emph{Anchor Graph Hashing} (AGH) \cite{liu2011hashing}, \emph{Multidimensional Spectral Hashing} (MDSH) \cite{weiss2012multidimensional}, and \emph{Scalable Graph Hashing} (SGH) \cite{jiang2015scalable} define notions of \emph{similarity} based on a function of $\ell_2$-distance between the data points ${\| \vecz_i - \vecz_j \|}_2$ and use spectral methods to learn hash functions that keep similar points sufficiently close.
Some other hashing algorithms first project the data onto its principal components, e.g., \emph{PCA Hashing} (PCAH) \cite{jolliffe2002principal}, which embeds with minimal average distortion, and then learn a rotation matrix to minimize the quantization error \cite{gong2011iterative} or balance the variance across the components (\emph{Isotropic Hashing} (IsoHash) \cite{kong2012isotropic}).

While minimizing the average distortion seems natural, this approach can sacrifice the preservation of certain pairwise distances in favor of others. 
As we demonstrate below, this can lead to poor performance in certain applications, such as the preservation of nearest neighbors.
In response, in this paper, we develop a new data-driven hashing algorithm that minimizes the \emph{worst-case} distortion among the pairwise distances, i.e., the $\ell_{\infty}$-norm of the pairwise-distortion vector.

Figure~\ref{fig:LinfL2} illustrates the advantages of minimizing the $\ell_{\infty}$-norm of the pairwise-distortion vector instead of its $\ell_2$-norm.
Consider three clusters of points in a two-dimensional space.
We compute the optimal one-dimensional embeddings of the data points by minimizing the $\ell_\infty$-norm and the $\ell_2$-norm of the pairwise-distortion vector using a grid search over the angular orientation of the line that represents the embedding.
We evaluate the near-neighbor preservation of a given query point in the embedded space (shown in \fref{fig:LinfL2}~(b)). 
For a query point $q$, located without loss of generality at the origin, the nearest neighbor ranking from the ambient space is destroyed by the $\ell_2$-optimal embedding, since the circle and square clusters overlap.
In contrast, the $\ell_{\infty}$-optimal embedding exactly preserves the rankings.

This illustration emphasizes the importance of preserving \emph{relevant} distances in data retrieval tasks. To minimize the average distortion, the $\ell_2$-optimal embedding (dashed red line) sacrifices the square--circle distances in favor of the square--star and circle--star distances, which contribute more to the $\ell_2$-norm of the pairwise distance distortion vector.
In contrast, the $\ell_{\infty}$-optimal embedding focuses on the hardest distances to preserve (i.e., the worst-case distortion), leading to an embedding with smaller isometry constant than the $\ell_2$-optimal embedding. 
Preservation of these distances is critical for near-neighbor retrieval.

\begin{figure}[t]
\vspace{-0.1cm}
\centering
\includegraphics[width=0.45\textwidth]{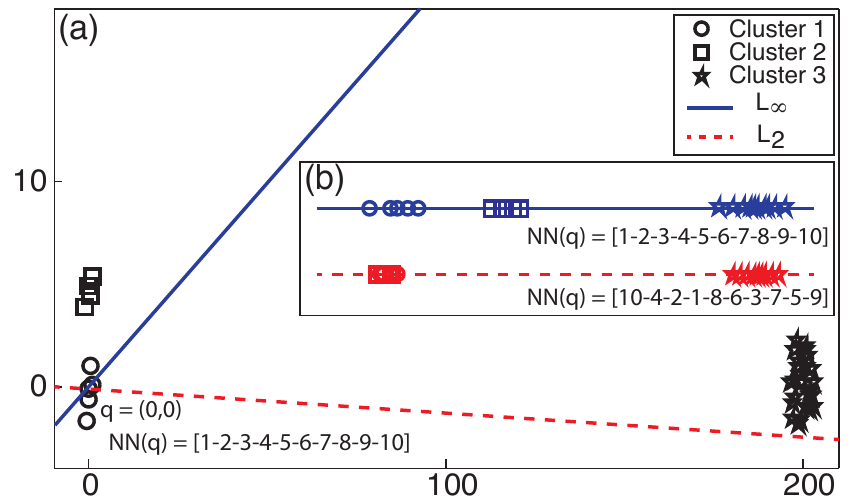}\label{fig:LinfL2}
\vspace{-.1cm}
\caption{Comparison of the near-neighbor (NN) preservation performance of hashing based on minimizing the $\ell_{\infty}$-norm (worst-case distortion) vs.\ the $\ell_2$-norm (average distortion) of the pairwise distance distortion vector on an illustrative data set. (a)~For a dataset with three clusters (5 circles, 5 squares, and 60 stars), we found the optimal embeddings for both error metrics using grid search. (b)~The projection of the data points using the $\ell_{\infty}$-optimal embedding preserves three well-separated clusters; however the projection using the $\ell_2$-optimal embedding mixes circles and squares, projecting them into a single cluster.  For the query point $\text{q}=(0,0)$, all of its nearest neighbors NN(q) are preserved with the correct ordering using the worst-based distortion embedding but not the average distortion embedding.
}
\label{fig:str}
\vspace{-.5cm}
\end{figure}

\vspace{-0.2cm}
\subsection{Contributions}
\label{sec:contributions}

We make four distinct contributions in this paper.  
First, conceptually, we advocate minimizing the worst-case distortion, which is formulated as an $\ell_\infty$-norm minimization problem, and show that this approach outperforms approaches based on minimizing the average, $\ell_2$-norm distortion in a range of computer vision and learning scenarios \cite{linfvsl2}. 

Second, algorithmically, since $\ell_\infty$-norm minimization problems are computationally challenging, especially for large datasets, we develop two accelerated and scalable algorithms to find the optimal worst-case embedding.
The first, \emph{near-isometric binary hashing} (NIBH), is based on the alternating direction method of multipliers (ADMM) framework \cite{boydadmm}. 
The second, NIBH-CG, is based on an accelerated greedy extension of the NIBH algorithm using the concept of \emph{column generation} \cite{dantzig1960decomposition}.
NIBH-CG can rapidly learn hashing functions from large-scale data sets that require the preservation of \emph{billions} of pairwise distances (e.g., \emph{MNIST}).

Third, theoretically, since current data-dependent hashing algorithms do not offer any probabilistic guarantees in terms of preserving near-neighbors, we develop new theory to prove that, under natural assumptions regarding the data distribution and with a notion of hardness of near-neighbor search, NIBH preserves the nearest neighbors with high probability.  
Our analysis approach could be of independent interest for obtaining theoretical guarantees for other data-dependent hashing schemes. 

Fourth, experimentally, we demonstrate the superior performance of NIBH as compared to ten state-of-the-art binary hashing algorithms using an exhaustive set of experimental evaluations involving six diverse datasets and three different performance metrics (near-isometry, Hamming distance ranking, and kendall $\tau$ ranking performance). 
In particular, we show that NIBH achieves the same distance preservation and Hamming ranking performance as state-of-the-art algorithms {\em while using up to $60\%$ fewer bits.}
Our experiments clearly show the superiority of the $\ell_\infty$-norm formulation over the more classical $\ell_2$-norm formulation that underlies many hashing algorithms, such as BRE and IsoHash. 
Our formulation also outperforms recently developed techniques that assume more structure in their hash functions, such as \emph{Spherical Hamming Distance Hashing} (SHD) \cite{heo2012spherical} and \emph{Circulant Binary Embedding} (CBE) \cite{yu2014circulant}.
\vspace{-0.2cm}

\fussy

\section{Near-Isometric Binary Hashing} 
\label{sec:algos} 

The standard formulation for data dependent binary hash function embeds a data point $\vecx\in\mathbb{R}^N$ into the low-{\break}dimensional Hamming space $\mathcal{H}=\{0,1\}^M$ by first multiplying it by an {\em embedding matrix} $\bW\in\mathbb{R}^{M\times N}$ and then quantizing the entries of the product $\bW\vecx$ to binary values:
\begin{align} \label{eq:hashfn}
h(\bW \vecx) = \frac{1+\text{sgn}(\bW \vecx)}{2}.
\end{align}
The function $\text{sgn}(\cdot)$ operates element-wise on the entries of $\bW \vecx$, transforming the real-valued vector $\bW \vecx$ into a set of binary codes depending on the sign of the entries in $\bW \vecx$.

\subsection{Problem formulation}
\label{sec:prob}

\sloppy
Consider the design of an embedding $f$ that maps $Q$ high-dimensional data vectors $\mathcal{X}=\{\vecx_1,\vecx_2, \dots, \allowbreak \vecx_Q\}$ in the ambient space $\mathbb{R}^N$ into low-dimensional binary codes  $\mathcal{H} = \{\vech_1,\vech_2,\dots,\vech_Q\}$ in the Hamming space with $\vech_i \in \{0,1\}^M$, where $\vech_i = f(\vecx_i)$, $i=1,\dots,Q$, and $M \ll N$.
Define the distortion of the embedding by 
\begin{align*}
\delta= & \underset{\lambda>0}{\text{inf}} \sup_{(i,j)\in \Omega} |\lambda {d}_{H}(\vech_i,\vech_j) - {d}(\vecx_i,\vecx_j)|, \\
& \text{with} \quad \Omega =  \{(i,j):i,j \in\{1,2,\dots,Q\},i>j \},
\end{align*}
where $d(\vecx_i,\vecx_j)$ denotes the Euclidean distance between the data points $\vecx_i$, $\vecx_j$, ${d}_{H}(\vech_i,\vech_j)$ denotes the Hamming distance between the binary codes $\vech_i$ and $\vech_j$, and $\lambda$ is a positive scaling variable. 
The distortion $\delta$ measures the worst-case deviation from perfect isometry (i.e., optimal distance preservation) among all pairs of data points.
Define the \emph{secant set} $\mathcal{S}(\mathcal{X})$ as $\mathcal{S}(\mathcal{X}) = \{ \vecx_i - \vecx_j : (i,j) \in \Omega \}$, i.e., the set of all pairwise difference vectors in $\mathcal{X}$.
Let $|\mathcal{S}(\mathcal{X})| = |\Omega| = {Q(Q-1)}/{2}$ denote the size of the secant set. 
Note that the common distortion measure utilized in other hashing algorithms is the average distortion, i.e., $\sum_{i>j}  (\lambda {d}_{H}(\vech_i,\vech_j) - {d}(\vecx_i,\vecx_j))^2/|\Omega|$.

We formulate the problem of minimizing the distortion parameter $\delta$ as the following optimization problem: 
\begin{align*}
\underset{\bW,\lambda > 0}{\text{minimize}} \underset{(i,j)\in \Omega}{\sup} \abs{ \lambda \normtwo{h(\bW \vecx_i) - h(\bW \vecx_j)}^2 - \normtwo{\vecx_i - \vecx_j} },
\end{align*}
since the squared $\elltwo$-distance between a pair of binary codes is equivalent to their Hamming distance up to a scaling factor that can be absorbed into $\lambda$.
We can rewrite the above optimization problem as 
\begin{align*} 
(\text{P}^*) \quad \underset{\bW,\lambda>0}{\text{minimize}}\,\,\, \| \lambda \vecv'(\bW)  - \vecc \|_\infty,
\end{align*}
\sloppy
where $\vecv' \in \mathbb{N}^{{Q(Q-1)}/{2}}$ is a vector containing the pairwise Hamming distances between the embedded data vectors $\normtwo{h(\vecx_i) - h(\vecx_j)}^2$, and $\vecc$ is a vector containing the pairwise $\elltwo$-distances between the original data vectors. 
Intuitively, the $\ell_{\infty}$-norm objective optimizes the \emph{worst-case} distortion among all pairs of data points.

The problem $(\text{P}^*)$ is a combinatorial problem with complexity $\mathcal{O} (Q^{2M})$.
To overcome the combinatorial nature of the problem, we approximate the hash function $h(\cdot)$ by the sigmoid function (also known as the inverse logit link function) $\sigma(x)=(1+e^{-x})^{-1}$. 
This enables us to approximate ($\text{P}^*$) by the following optimization problem:
\begin{align*} 
(\text{P}) \quad \underset{\bW,\lambda>0}{\text{minimize}}\,\,\, \| \lambda \vecv(\bW)  - \vecc \|_\infty,
\end{align*}
where $\vecv \in \mathbb{R}_+^{{Q(Q-1)}/{2}}$ is a vector containing the pairwise $\ell_2$ distances between the embedded data vectors after sigmoid relaxation $\normtwo{(1+e^{-\bW \vecx_i})^{-1} - (1+e^{-\bW \vecx_j})^{-1}}^2$. 
Here the sigmoid function operates element-wise on $\bW \vecx_i$. 
In practice we use a more general definition of the sigmoid function, defined as $\sigma_\alpha(x)=(1+e^{-\alpha x})^{-1}$, where $\alpha$ is the \emph{rate} parameter controlling how closely it approximates the non-smooth function $h(\cdot)$. The following lemma characterizes the quality of such an approximation (see the Appendix for a proof).

\begin{lem} \label{lem:approx}
Let $x$ be a Gaussian random variable as $x \sim \mathcal{N}(\mu, \sigma^2)$. Define the distortion of the sigmoid approximation at $x$ as $\abs{h(x)-\sigma_\alpha(x)}$. Then, the expected distortion is bounded as $\mathbb{E}_x[ \abs{h(x)-\sigma_\alpha(x)} ] \leq \frac{1}{\sigma\sqrt{2\pi\alpha}} + 2 e^{-(\sqrt{\alpha} + c /\alpha \sigma^2)}$,
where $c$ is a positive constant. As $\alpha$ goes to infinity, the expected distortion goes to $0$. 
\end{lem}

\begin{remark}
As has been noted in the machine vision literature \cite{zoran2012natural}, a natural model for an image database is that its images are generated from a mixture of Gaussian distributions. \fref{lem:approx} bounds the deviation of the sigmoid approximation from the non-smooth hash function \fref{eq:hashfn} under this model. 
\end{remark}

\subsection{Near-isometry and nearest neighbor preservation}
\label{sec:NNdelta}

Inspired by the definition of \emph{relative contrast} in \cite{he2012difficulty}, we define a more generalized measure of data separability to preserve $k$-NN that we call the {\em $k$-order gap} $\Delta_k := d(\vecx_0, \vecx_{k+1}) - d(\vecx_0, \vecx_k)$, where $\vecx_0$ is a query point and $\vecx_k$ and $\vecx_{k+1}$ are its $k^\text{th}$ and $k+1^\text{th}$ nearest neighbors, respectively. 
We formally show that if the data is highly separable ($\Delta_k$ is large), then the above approach preserves all $k$ nearest neighbors with high probability (see the Appendix for a proof and discussion).

\begin{thm} \label{thm:knn}
Assume that all the data points are independently generated from a mixture of Gaussian distribution i.e., $\vecx_i \sim \sum_{p=1}^P \pi_p \mathcal{N}(\mu_p,\Sigma_p)$. 
Let $\vecx_0 \in \mathbb{R}^N$ denote a query data point in the ambient space, and the other data points $\vecx_i$ be ordered so that $d(\vecx_0,\vecx_1) < d(\vecx_0,\vecx_2) < \ldots < d(\vecx_0,\vecx_Q)$.  Let $\delta$ denote the final value of the distortion parameter computed from any binary hashing algorithm, and let $c$ denote a positive constant. 
Then, if $\mathbb{E}_x[\Delta_k] \geq 2 \delta + \sqrt{\frac{1}{c}\log \frac{Qk}{\epsilon}}$, the binary hashing algorithm preserves all the $k$-nearest neighbors of a data point with probability at least $1-\epsilon$.
\end{thm}

\subsection{The NIBH algorithm}
\label{sec:admm}
We now develop an algorithm to solve the optimization problem (P). 
We apply the alternating direction method of multipliers (ADMM) framework \cite{boydadmm} to construct an efficient algorithm to find a (possibly local) optimal solution of (P).
Note that (P) is non-convex, and therefore no standard optimization method is guaranteed to converge to a globally optimal solution in general.
We introduce an auxiliary variable $\vecu$ to arrive at the equivalent problem:
\begin{align} \label{eq:eqprob} 
\underset{\bW,\vecu,\lambda>0}{\text{minimize}}\,\,\, \| \vecu  \|_\infty \quad \text{subject to} \quad \vecu = \lambda \vecv(\bW) - \vecc. 
\end{align}
The augmented Lagrangian form of this problem can be written as $\underset{\bW,\vecu,\lambda>0}{\text{minimize}}\,\,\, \| \vecu \|_\infty + \frac{\rho}{2} \normtwo{\vecu- \lambda \vecv(\bW) + \vecc + \vecy}^2,$
where $\rho$ is the scaling parameter in ADMM and $\vecy \in \mathbb{R}^{{Q(Q-1)}/{2}}$ is the Lagrange multiplier vector. 
The NIBH algorithm proceeds as follows. 
First, the variables $\bW$, $\lambda$, $\vecu$, and Lagrange multipliers $\vecy$ are initialized randomly. 
Then, at each iteration, we optimize over each of the variables ${\bW, \vecu,}$ and ${\lambda}$ while holding the other variables fixed. 
More specifically, in iteration $\ell$, we perform the following four steps until convergence:
\begin{itemize}[leftmargin=*]
\fussy
\item{\emph{Optimize over} $\bW$} via 
$\bW^{(\ell+1)}\!\! \leftarrow \!\! \underset{\bW}{\text{arg min}} \frac{1}{2} \sum_{(i,j)\in\Omega} ( u_{ij}^{(\ell)} - \lambda^{(\ell)} \| \frac{1}{1+e^{-\bW \vecx_i}} - \frac{1}{1+e^{-\bW \vecx_j}} \|_2^2 +\normtwo{\vecx_i - \vecx_j}^2 - y_{ij}^{(\ell)} )^2$,
where $\lambda^{(\ell)}$ denotes the value of $\lambda$ in the $\ell^\text{th}$ iteration. 
We also use $u_{ij}^{(\ell)}$ and $y_{ij}^{(\ell)}$ to denote the entries in $\vecu^{(\ell)}$ and $\vecy^{(\ell)}$ that correspond to the pair $\vecx_i$ and $\vecx_j$ in the dataset $\mathcal{X}$. 
We show in our experiments below that using the accelerated first-order gradient descent algorithm \cite{nest} to solve this subproblem results in good empirical convergence performance (see the Appendix). 
\sloppy

\item{\emph{Optimize over} $\vecu$} while holding the other variables fixed; it corresponds to solving the proximal problem of the $\ell_\infty$-norm $\vecu^{(\ell+1)} \!\!\leftarrow \!\! \underset{\vecu}{\text{arg min}} \, \| \vecu  \|_\infty  \! + \! \frac{\rho}{2} \| \vecu \! - \! \lambda^{(\ell)} \! \vecv^{(\ell+1)} \! + \vecc + \! \vecy^{(\ell)} \|_2^2$. 
We use the low-cost algorithm described in \cite{studertom} to perform the proximal operator update. 

\fussy
\item{\emph{Optimize over} $\lambda$} while holding the other variables fixed; it corresponds to a positive least squares problem, where $\lambda$ is updated as $\lambda^{(\ell+1)}\!\! \leftarrow \!\! \underset{\lambda>0}{\text{arg min}} \, \frac{1}{2} \| \vecu^{(\ell+1)} - \lambda \vecv^{(\ell+1)} + \vecc + \vecy^{(\ell)} \|_2^2. $
We perform this update using the non-negative least squares algorithm \cite{fcnnls}.

\item{\emph{Update} $\vecy$} via $\vecy^{(\ell+1)} \!\!  \leftarrow \!\!  \vecy^{(\ell)} + \eta (\vecu^{(\ell+1)} - \lambda^{(\ell+1)}\vecv^{(\ell+1)}+\vecc)$,
where the parameter $\eta$ controls the dual update step size. \fussy

\end{itemize} 
\subsection{Accelerated NIBH for large-scale datasets}
\label{sec:nibhcg}

The ADMM-based NIBH algorithm is efficient for small-scale datasets (e.g., for secant sets of size $|\mathcal{S}(\mathcal{X})| < 5000$ or so). 
However, the memory requirement of NIBH is quadratic in $|\mathcal{X}|$, which would be problematic for applications involving large-scale numbers of data points and secants. 
In response, we develop an algorithm that approximately solves $(\text{P})$ while scaling very well to large-scale problems. 
The key idea comes from classical results in optimization theory related to {\em column generation} (CG) \cite{dantzig1960decomposition,numax}. 

The optimization problem \fref{eq:eqprob} is an $\ell_{\infty}$-norm minimization problem with an equality constraint on each secant. 
The Karush-Kuhn-Tucker (KKT) condition for this problem states that, if strong duality holds, then the optimal solution is entirely specified by a (typically very small) portion of the constraints. 
Intuitively, the secants corresponding to these constraints are the pairwise distances that are harder to preserve in the low-dimensional Hamming space. 
We call the set of such secants the \emph{active} set. 
In order to solve $(\text{P})$, it suffices to find the active secants and solve NIBH with a much smaller number of active constraints.
To leverage the idea of the active set, we iteratively run NIBH on a small subset of all the secants that violate the near-isometry condition, as detailed below:

\begin{itemize}

\item 
Solve $(\text{P})$ with a small random subset $\mathcal{S}_0$ of all the secants $\mathcal{S}(\mathcal{X})$ using NIBH to obtain $\widehat{\bW}$, $\hat{\delta}$, and $\hat{\lambda}$, initial estimates of the parameters. Identify the active set $\mathcal{S}_\text{a}$. Fix $\lambda=\hat{\lambda}$ for the rest of the algorithm.
\item 
Randomly select a new subset $\mathcal{S}_\text{v} \subset \mathcal{S}$ of secants that violate the near isometry condition using the current estimates of $\widehat{\bW}$, $\hat{\delta}$, and $\hat{\lambda}$. Then, form an augmented secant set $\mathcal{S}_\text{aug} = \mathcal{S}_\text{a} \cup  \mathcal{S}_\text{v}$. 
\item 
Solve $(\text{P})$ with the secants in the set $\mathcal{S}_\text{aug}$ using the NIBH algorithm.

\end{itemize}

We dub this approach {\em NIBH-CG}.
NIBH-CG iterates over the above steps until no new violating secants are added to the active set.
Since the algorithm searches over all the secants for violating secants in each iteration before terminating, NIBH-CG ensures that all of the constraints are satisfied when it terminates. 
A key benefit of NIBH-CG is that only the set of active secants (and not all secants) needs to be stored in memory. 
This benefit leads to significant improvements in terms of memory complexity over competing algorithms, since the set of all secants quickly becomes large-scale and can exceed the system memory capacity in large-scale applications.


\section{Experiments}
\label{sec:experiments}

\begin{figure*}[t]
\vspace{-0.1cm}
\centering
\includegraphics[width=1\textwidth]{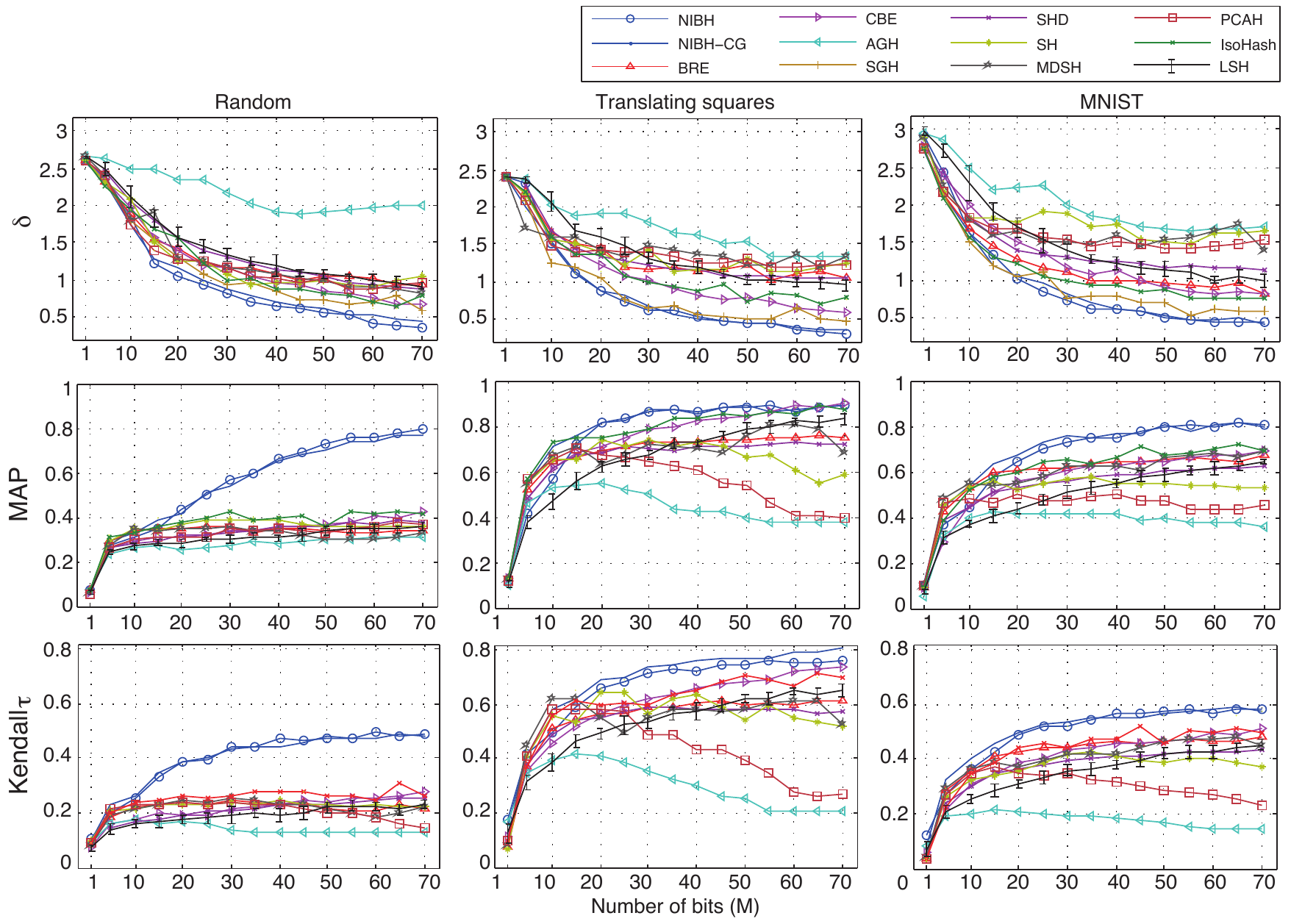}\label{fig:str}
\vspace{-.5cm}
\caption{Comparison of the NIBH and NIBH-CG algorithms against several baseline binary hashing algorithms using three small-scale datasets with 4950 secants ($Q$ = 100). The performance of NIBH-CG closely follows that of NIBH, and both outperform all of the other algorithms in terms of the maximum distortion $\delta$ (superior distance preservation), mean average precision $\text{MAP}$ of training samples (superior nearest neighbor preservation), and Kendall $\tau$ rank correlation coefficient (superior ranking preservation).}
\label{fig:str}
\vspace{-.2cm}
\end{figure*}

\label{fig:deltatest}

In this section, we validate the NIBH and NIBH-CG algorithms via experiments using a range of synthetic and real-world datasets, including three small-scale, three medium-scale, and one large-scale datasets with respect to three metrics. 
We compare NIBH against ten state-of-the-art binary hashing algorithms, including
binary reconstructive embedding (BRE) \cite{kulis2009learning},
spectral hashing (SH) \cite{weiss2009spectral},
anchor graph hashing (AGH) \cite{liu2011hashing}, 
multidimensional spectral hashing (MDSH) \cite{weiss2012multidimensional},
scalable graph hashing (SGH) \cite{jiang2015scalable},
PCA hashing (PCAH) \cite{jolliffe2002principal},
isotropic hashing (IsoHash) \cite{kong2012isotropic},
spherical Hamming distance hashing (SHD) \cite{heo2012spherical}, 
circulant binary embedding (CBE) \cite{yu2014circulant},
and
locality-sensitive hashing (LSH) \cite{indyk1998approximate}.

\begin{figure*}[tp]
\centering
\includegraphics[width=0.8\textwidth]{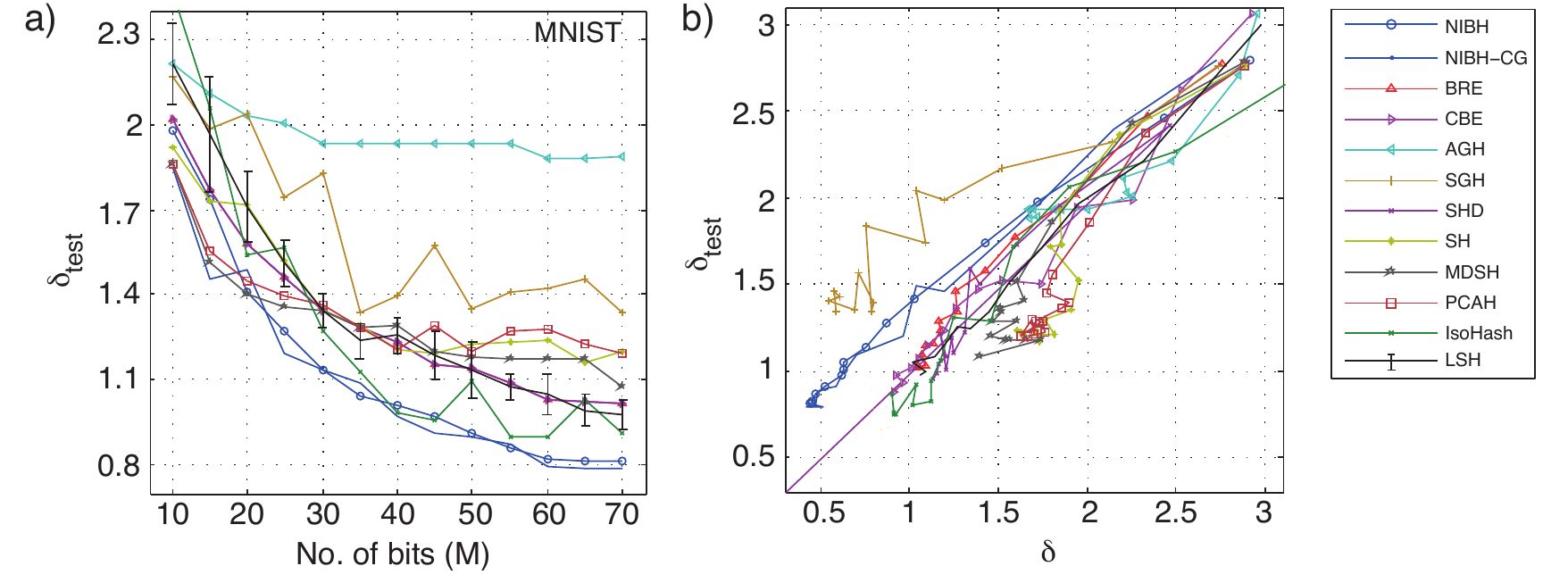}\label{fig:deltatest}
\caption{Comparison of NIBH and NIBH-CG against several state-of-the-art binary hashing algorithms in preserving isometry on MNIST data. (a) NIBH-CG outperforms the other algorithms in minimizes the isometry constant on unseen data $\delta_{\text{test}}$. (b)  . NIBH and NIBH-CG provide better isometry guarantee with a small sacrifice to universality.}
\end{figure*}

\subsection{Performance metrics and datasets}
\label{sec:perfmetric}

We compare the algorithms using the following three metrics: 

{\emph{Maximum distortion}} $\delta= \underset{\lambda>0}{\text{inf}} ||\lambda \hat{\vecv} - \vecc||_{\infty}$, where the vector $\hat{\vecv}$ contains the pairwise Hamming distances between the learned binary codes. 
This metric quantifies the distance preservation among all of the pairwise distances after projecting the training data in the ambient space into binary codes. 
We also define the maximum distortion for unseen test data $\delta_{\text{test}}$, which measures the distance preservation on a hold-out test dataset using the hash function learned from the training dataset. 

{\emph{Mean average precision}} (MAP) for near-neighbor preservation in the Hamming space. MAP is computed by first finding the set of $k$-nearest neighbors for each query point on a hold-out test data in the ambient space $\mathcal{L}^k$ and the corresponding set $\mathcal{L}_\text{H}^k$ in the Hamming space and then calculating the average precision $\text{AP} = |\mathcal{L}^k \cap \mathcal{L}_\text{H}^k |/k$. 
We then report MAP by calculating the mean value of $\text{AP}$ across all data points. 

{\emph{Kendall $\tau$ ranking correlation coefficient}}. We first rank the set of $k$-nearest neighbors for each data point by increasing distance in the ambient space as $\mathcal{T}(\mathcal{L}^k)$ and in the Hamming space as $\mathcal{T}(\mathcal{L}_\text{H}^k)$.
The Kendall $\tau$ correlation coefficient is a scalar $\tau\in[-1,1]$ that measures the similarity between the two ranked sets $\mathcal{T}(\mathcal{L}^k)$ and $\mathcal{T}(\mathcal{L}_\text{H}^k)$ \cite{kendall1938new}. 
The value of $\tau$ increases as the similarity between the two rankings increases and reaches the maximum value of $\tau=1$ when they are identical. 
We report the average value of $\tau$ across all data points in the training dataset.

To compare the algorithms, we use the following standard datasets from computer vision:  
\emph{Random} consists of independently drawn random vectors in $\mathbb{R}^{100}$ from a multivariate Gaussian distribution with zero mean and identity covariance matrix. 
 \emph{Translating squares} is a synthetic dataset consisting of $10 \times 10$ images that are translations of a $3\times 3$ white square on black background \cite{numax}.
\emph{MNIST} is a collection of 60,000 $28 \times 28$ greyscale images of handwritten digits  
\cite{lecun1998mnist}.
\emph{Photo-Tourism} is a corpus of approximately 300,000 image patches, represented using scale-invariant feature transform (SIFT) features \cite{lowe2004distinctive} in $\mathbb{R}^{128}$ \cite{snavely2006photo}.
\emph{LabelMe} is a collection of over 20,000 images represented using GIST descriptors in $\mathbb{R}^{512}$ \cite{torralba2008small}. 
\emph{Peekaboom} is a collection of 60,000 images represented using GIST descriptors in $\mathbb{R}^{512}$ \cite{torralba2008small}.
Following the experimental approaches of the hashing literature \cite{kulis2009learning,norouzi2011minimal}, we pre-process the data by subtracting the mean and then normalizing all points to lie on the unit sphere. 
\vspace{-0cm}
\subsection{Small- and medium-scale experiments}
\label{sec:smalld}

We start by evaluating the performance of NIBH and NIBH-CG using a small-scale subset of the first three datasets.
Small-scale datasets enable us to compare the performance of NIBH vs.\ NIBH-CG to verify that they perform similarly.
Also they help us assess the asymptotic behavior of algorithms in preserving isometry since the total of number of secants are small compare to the bit budget in compact binary codes.

\paragraph{Experimental setup.}
We randomly select $Q=$ 100 data points from the \emph{Random}, \emph{Translating squares}, and \emph{MNIST} datasets.
We then apply the NIBH, NIBH-CG, and all the baseline algorithms on each dataset for different target binary code word lengths $M$ from 1 to 70 bits. 
We set the NIBH and NIBH-CG algorithm parameters to the common choice of $\rho=1$ and $\eta=1.6$. 
To generate hash function of length $M$ for LSH, we draw $M$ random vectors from a Gaussian distribution with zero mean and an identity covariance matrix. We use the same random vectors to initialize NIBH and other baseline algorithms.
In the near-neighbor preservation experiments, to show the direct advantage of minimizing $\ell_{\infty}$-norm over $\ell_2$-norm, we followed the exact procedure described in BRE \cite{kulis2009learning} to select the training secants, i.e., we apply the NIBH algorithm on only the lowest $5\%$ of the pairwise distances (which are set to zero as in BRE) combined with the highest $2\%$ of the pairwise distances.

We follow the \emph{continuation} approach \cite{wen2010fast} to set the value of $\alpha$. 
We start with a small value of $\alpha$, (e.g., $\alpha=1$) to avoid becoming stuck in bad local minima, and then gradually increase $\alpha$ as the algorithm proceeds. 
As the algorithm gets closer to convergence and has obtained a reasonably good estimate of the parameters $\bW$ and $\lambda$, we set $\alpha = 10$, which enforces a good approximation of the sign function (see Lemma \ref{lem:approx}). 

\paragraph{Results.}
\vspace{-0.2cm}

The plots in the top row of Figure~\ref{fig:str} illustrate the value of the distortion parameter $\delta$ as a function of the number of projections (bits) $M$. 
The performance of NIBH and NIBH-CG closely follow each other, indicating that NIBH-CG is a good approximation to NIBH.
Both NIBH and NIBH-CG outperform the other baseline algorithms in terms of the distortion parameter $\delta$. 
Among these baselines, LSH has the lowest isometry performance since random projections are oblivious to the intrinsic geometry of the training dataset.
To achieve $\delta=$ 1, NIBH(-CG) requires $60\%$ fewer bits $M$ than CBE and BRE.
NIBH(-CG) also achieves better isometry performance asymptotically, i.e., up to $\delta \approx$ 0.5, given a sufficient number of bits ($M\geq$ 70), while for most of the other algorithms the performance plateaus after $\delta=$ 1.
NIBH's superior near-isometry performance extends well to unseen data. 
Figure \ref{fig:deltatest}(a) demonstrates that NIBH achieves the lowest isometry constant on a test dataset $\delta_{\text{test}}$ compared to other hashing algorithms.
Figure \ref{fig:deltatest}(b) further suggests that NIBH's superior isometry performance comes with smallest sacrifice to the universality of the hash functions.

\begin{figure*}[tp]
\vspace{-0.1cm}
\centering
\includegraphics[width=0.9\textwidth]{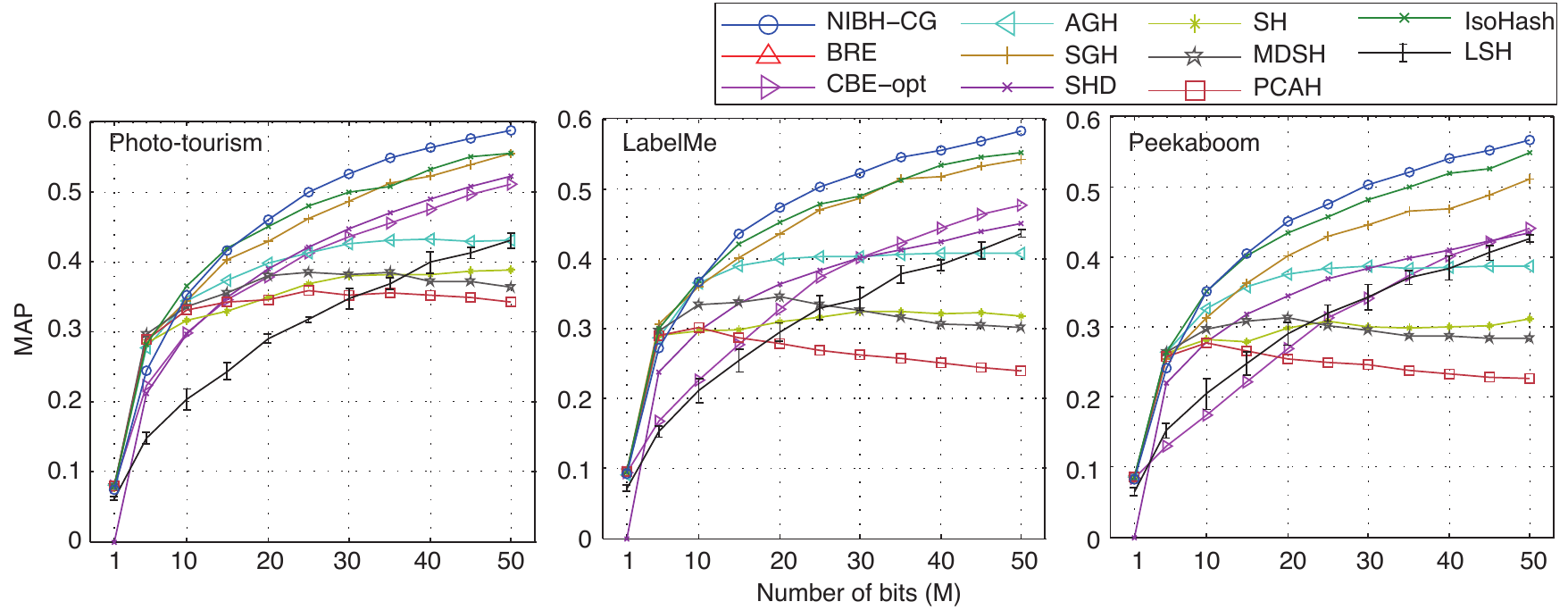}\label{fig:btr}
\caption{Hamming ranking performance comparison on three medium-scale datasets ($Q=$ 1000). The top-50 neighbors are used to report MAP over a test data of same size.}
\label{fig:btr}
\vspace{-0.5cm}
\end{figure*} 

The plots in the middle and bottom row of Figure~\ref{fig:str} shows the average precision for retrieving training data and the Kendall $\tau$ correlation coefficient respectively, as a function of the number of bits $M$.
We see that NIBH preserves a higher percentage of nearest neighbors compared to other baseline algorithms as $M$ increases with better average ranking  among $k=$ 10-nearest neighbors.

Now we showcase the performance of NIBH-CG on three medium-scale, real-world datasets used in \cite{kulis2009learning,norouzi2011minimal}, including \emph{Photo-tourism}, \emph{LabelMe}, and \emph{Peekaboom} for the popular machine learning task of {\em data retrieval}.
From each dataset we randomly select $Q=$ 1000 training points, following the setup in BRE \cite{kulis2009learning}, and use them to train NIBH-CG and the other baseline algorithms.
We then randomly select a separate set of $Q=$ 1000 data points and use it to test the performance of NIBH-CG and other baseline algorithms in terms MAP with $k = 50$.
Figure~\ref{fig:btr} illustrates the performance of NIBH-CG on these datasets. 
NIBH-CG outperforms all the baseline algorithms with large margins in Hamming ranking performance in term of MAP with top-50 near-neighbors.

\subsection{Large-scale experiments}
\label{sec:larged}

We now demonstrate that NIBH-CG scales well to large-scale datasets.
We use the full \emph{MNIST} dataset with 60,000 training images and augment it with three rotated versions of each image (rotations of 90$^{\circ}$, 180$^{\circ}$, and 270$^{\circ}$) to create a larger dataset with $Q=$ 240,000 data points. 
Next, we construct 4 training sets with 1,000, 10,000, 100,000, and 240,000 images out of this large set.
We train all algorithms with $M =$ 30 bits and compare their performance on a test set of 10,000 images. 
BRE fails to execute on a standard desktop PC with 12 GB of RAM for training sets with more than 100,000 points due to the size of the secant set $|\mathcal{X}|$.
The results for all algorithms are given in Table~\ref{tab:largescale}; we tabulate their performance in terms of MAP for the top-500 neighbors. 
The performance of NIBH-CG is significantly better than the baseline algorithms and, moreover, improves as the size of the training set grows. 
This emphasizes that NIBH-CG excels at large-scaled problems thanks to its very small memory requirement; indeed, the memory requirement of NIBH-CG is linear in the number of \emph{active} secants rather than the total number of secants.

\vspace{-0.4cm}
\section{Discussion}

We have demonstrated that the worst-case, $\ell_{\infty}$-norm-based near-isometric binary hashing (NIBH) algorithm is superior to a wide range of algorithms based on the more traditional average-case, $\ell_2$-norm. 
Despite its non-convexity and non-smoothness, NIBH admits an efficient optimization algorithm that converges to a high-performing local minimum.
Moreover, NIBH-CG, the accelerated version of NIBH, provides significant memory advantages over existing algorithms. 
Our exhaustive experiments with six datasets, three metrics, and ten algorithms have shown that NIBH outperforms all of the state-of-the-art data-dependent hashing algorithms. 
The results in this paper provide a strong motivation for exploring $\ell_{\infty}$-norm formulations in binary hashing.  

\begin{table*}
\centering
\caption{Comparison of NIBH-CG against several baseline binary hashing algorithms on large-scale  \emph{MNIST} datasets with over 28 billion secants $|\mathcal{S}(\mathcal{X})|$.  We tabulate the Hamming ranking performance in terms of mean average precision MAP for different sizes of the dataset. All training times are in seconds.}
\scalebox{0.8}{
\label{tab:largescale}
\begin{tabular}{|c|c|c|c|c|c|}
  \hline  $M = 30$ \text{bits}  & \multicolumn{4}{c|}{  MAP / Top-500 ( \emph{MNIST + rotations})} &  \multicolumn{1}{c|}{ Training time } \\
  \hline  {Training size} $Q$ & $1K$ & $10K$ & $100K$ & $240K$ & $240K$\\
  \hline  {Secant size} $|\mathcal{S}(\mathcal{X})|$ &  $500K$  & $50M$ & $5B$ & $28B$ & $28B$ \\
  \hline 
  \hline \textbf{NIBH-CG} & $\bf{52.79}$ ($\pm 0.15$)  & $\bf{54.69}$ ($\pm 0.18$) & $\bf{54.93}$ ($\pm 0.23$) &  $\bf{55.52}$ ($\pm 0.11$) & 541.43 \\
  \hline BRE   & $48.33$ ($\pm 0.65$) & $50.67$($\pm 0.33$) & -- & -- &  $18685.51$ \\
  \hline CBE  & $38.70$ ($\pm 1.18$) & $38.12$ ($\pm 1.34$) & $38.50$ ($\pm 2.05$) & $38.53$ ($\pm 0.83$) & 68.94\\
  \hline SPH  & $44.33$ ($\pm 0.74$) & $44.24$ ($\pm 0.61$) & $44.37$ ($\pm 0.71$) & $44.32$ ($\pm 0.63$) & 184.46  \\
  \hline SH  & $40.12$ ($\pm 0.00$) & $39.37$ ($\pm 0.00$) & $38.79$ ($\pm 0.00$) & $38.26$ ($\pm 0.00$) & 3.05 \\
  \hline MDSH  & $41.06$ ($\pm 0.00$) & $41.23$ ($\pm 0.00$) & $40.80$ ($\pm 0.00$) & $40.39$ ($\pm 0.00$) & 15.00 \\
  \hline AGH  & $45.81$ ($\pm 0.34$) & $47.78$ ($\pm$ 0.38) & $47.69$ ($\pm 0.41$) & $47.38$ ($\pm 0.32$) & 4.49 \\
  \hline SGH  & $51.32$ ($\pm 0.07$) & $51.33$ ($\pm 0.20$) & $51.01$ ($\pm 0.23$) & $50.66$ ($\pm 0.76$) & $5.89$ \\
  \hline PCAH  & $39.90$ ($\pm 0.00$) & $38.53$ ($\pm 0.00$) & $38.81$ ($\pm 0.00$) & $37.50$ ($\pm 0.00$) & $0.08$\\
  \hline IsoHash  & $50.91$ ($\pm 0.00 $) & $50.90$ ($\pm 0.00 $) & $50.72$ ($\pm 0.00 $) & $50.55$ ($\pm 0.00 $) & $2.82$ \\
  \hline LSH  & $33.69$ ($\pm 0.94 $) & $33.69$ ($\pm 0.94 $) & $33.69$ ($\pm 0.94 $) & $33.69$ ($\pm 0.94 $) & $2.29\times 10^{-4}$  \\
  \hline
\end{tabular}}
\end{table*}


\bibliography{icml2016}
\bibliographystyle{icml2016}

\newpage

\appendix

\section{Appendix}
\label{sec:appendix}

In the appendix, we prove \fref{lem:approx} and \fref{thm:knn} on the performance of NIBH on $k$-nearest neighbor preservation. Finally, we include additional numerical simulation results and discussions on the empirical convergence of the NIBH algorithm.

\section*{Proof of \fref{lem:approx}}
\begin{lem} \label{lem:approx}
Let $x$ be a Gaussian random variable as $x \sim \mathcal{N}(\mu, \sigma^2)$. Define the distortion of the sigmoid approximation at $x$ as $\abs{h(x)-\sigma_\alpha(x)}$. Then, the expected distortion is bounded as
\begin{align*} 
\mathbb{E}_x[ \abs{h(x)-\sigma_\alpha(x)} ] \leq \frac{1}{\sigma\sqrt{2\pi\alpha}} + 2 e^{-(\sqrt{\alpha} + c /\alpha \sigma^2)},
\end{align*} 
where $c$ is a positive constant. As $\alpha$ goes to infinity, the expected distortion goes to $0$.
\end{lem}
\sloppy
\begin{proof}
It is easy to see that the distortion $\abs{h(x)-\sigma_\alpha(x)}$ occurs at $x = 0$. Therefore, among different values of $\mu$, $\mu = 0$ gives the largest distortion since the density of $x$ peaks at $x = 0$. Therefore, we bound the distortion at setting $\mu = 0$, which is an upper bound of the distortion when $\mu \neq 0$. 
By definition (1) in the main text, $h(x)$ can be written as
\begin{align*}  
h(x) = \left \{ \begin{array}{ll}
1 &\text{if} \,\, x \geq 0 \\
0 & \text{otherwise}.
\end{array} \right.
\end{align*} 
When $x \sim \mathcal{N}(0,\sigma^2)$, we have
\begin{align}
\notag & \mathbb{E}_x[ \abs{h(x)-\sigma_\alpha(x)} ]  = \int_{-\infty}^\infty \abs{h(x)-\sigma_\alpha(x)} \mathcal{N}(x;0,\sigma^2) \mathrm{d}x \\
\label{eq:sym} & = 2 \int_0^\infty (h(x)-\sigma_\alpha(x)) \mathcal{N}(x;0,\sigma^2) \mathrm{d}x \\
\notag & = 2 \int_0^{x_0} (h(x)-\sigma_\alpha(x)) \mathcal{N}(x;0,\sigma^2) \mathrm{d}x \\
\notag & \quad + 2 \int_{x_0}^{\infty} (h(x)-\sigma_\alpha(x)) \mathcal{N}(x;0,\sigma^2) \mathrm{d}x \\
\label{eq:leq} & \leq \! 2 \! \int_0^{x_0} \frac{1}{2} \mathcal{N}(x;0,\sigma^2) \mathrm{d}x \! + \! 2 \! \int_{x_0}^{\infty} \frac{1}{1+e^{\alpha x_0}} \mathcal{N}(x;0,\sigma^2) \mathrm{d}x \\ 
\label{eq:apr} & \leq  \frac{x_0}{\sqrt{2\pi} \sigma} + 2 \frac{e^{-c x_0^2/\sigma^2}}{1+e^{\alpha x_0}} \\ 
\label{eq:den} & \leq \frac{1}{\sigma \sqrt{2\pi\alpha}} + 2 e^{-(\sqrt{\alpha} + c /\alpha\sigma^2)} ,
\end{align}
when we set $x_0 = \frac{1}{\sqrt{\alpha}}$ and $c$ is a positive constant. In \fref{eq:sym}, we used the fact that $\sigma_\alpha(x)$ and $h(x)$ are symmetric with respect to the point $(0,\frac{1}{2})$. \fref{eq:leq} is given by the properties of the sigmoid function, \fref{eq:apr} is given by the Gaussian concentration inequality \cite{talagrand}, and \fref{eq:den} is given by the inequality $1/(1+e^{\alpha x_0}) \leq e^{-\alpha x_0}$. The fact that $\mathbb{E}_x[ \abs{h(x)-\sigma_\alpha(x)} ] \rightarrow 0$ as $\alpha \rightarrow \infty$ is obvious from the bound above. 
\end{proof}

\section*{Proof of \fref{thm:knn}}

\begin{thm} \label{thm:knn}
Assume that all the data points are independently generated from a mixture of Gaussian distribution, i.e., $\vecx_i \sim \sum_{p=1}^P \pi_p \mathcal{N}(\mu_p,\Sigma_p)$. 
Let $\vecx_0 \in \mathbb{R}^N$ denote a query data point in the ambient space, and the other data points $\vecx_i$ be ordered so that $d(\vecx_0,\vecx_1) < d(\vecx_0,\vecx_2) < \ldots < d(\vecx_0,\vecx_Q)$.  Let $\delta$ denote the final value of the distortion parameter computed from any binary hashing algorithm, and let $c$ denote a positive constant. 
Then, if $\mathbb{E}_x[\Delta_k] \geq 2 \delta + \sqrt{\frac{1}{c}\log \frac{Qk}{\epsilon}}$, the binary hashing algorithm preserves the $k$-nearest neighbors of a point with probability at least $1-\epsilon$.
\end{thm}
In order to prove this theorem, we need the following Lemma:
\begin{lem} \label{lem:sg}
Let $\vecx_0, \ldots, \vecx_N$ and $\Delta_k$ be defined as in \fref{thm:knn}. Then, there exist a constant $c$ such that $P(\Delta_k - \mathbb{E}_x[\Delta_k] < t) \leq e^{-ct^2}$ for $t > 0$. 
\end{lem}
\begin{proof}
Since the data points $\vecx_0$, $\vecx_k$ and $\vecx_{k+1}$ are independently generated from a finite mixture of Gaussian distributions, the random variable of their concatenation $\vecy = [\vecx_0^T, \vecx_k^T, \vecx_{k+1}^T]^T \in \mathbb{R}^{3N}$ is sub-Gaussian \cite{wainwrightpaper}. Then, we have
\begin{align*}
\Delta_k(\vecy) & = \|\vecx_0 - \vecx_{k+1}\|_2 - \|\vecx_0 - \vecx_k\|_2 \\
& = \| \Big( \begin{array}{ccc} \bI & 0 & 0 \\ 0 & 0 & 0 \\ 0 & 0 & -\bI \end{array}\Big) \vecy \|_2 - \| \Big( \begin{array}{ccc} \bI & 0 & 0 \\ 0 & -\bI & 0 \\ 0 & 0 & 0 \end{array}\Big) \vecy \|_2 \\
& \leq  \| \Big( \begin{array}{ccc} 2\bI & 0 & 0 \\ 0 & -\bI & 0 \\ 0 & 0 & -\bI \end{array}\Big) \vecy \|_2 \\
& \leq 2 \| \vecy \|_2,
\end{align*}
where we have used the triangular inequality in the second to last step, the Rayleigh-Ritz theorem \cite{hornjohnson}, and the fact that the maximum singular value of the matrix in the step before is $2$. 
This result means that $\Delta_k(\vecy)$ is a Lipschitz function of $\vecy$. 
Thus, by Talagrand's inequality \cite{talagrand}, we have that $P(\Delta_k - \mathbb{E}_x[\Delta_k] < -t) \leq e^{-ct^2}$ for some positive constant $c$ and $t > 0$, since $\vecy$ is sub-Gaussian. 
\end{proof}
Now we are ready to prove \fref{thm:knn}. 
\begin{proof}
Let $E$ denote the event that the set of top-$k$ nearest neighbors is not preserved in the Hamming space. Then, we have $E = \cup e_{m,n}$, where $e_{m,n}$ denote the event that $d_H(\vecx_0,\vecx_m) > d_H(\vecx_0,\vecx_n)$ with $m \in \{1,\ldots,k\}$ and $n \in \{k+1,\ldots,Q\}$. 
Then, using the union bound \cite{infotheory}, we have
\begin{align*}
P(E) & \leq \sum_{m,n} P(e_{m,n}) \leq k(Q-k) P(e_{k,k+1}) \\
& \quad = k(Q-k) P(d_H(\vecx_0,\vecx_k) > d_H(\vecx_0,\vecx_{k+1})) \\
& \quad = k(Q-k) P(d_H(\vecx_0,\vecx_{k+1}) < d_H(\vecx_0,\vecx_k)),
\end{align*}
where we have used the fact that the most possible event among all $e_{m,n}$ events is the one corresponding to the order mismatch between the $k^\text{th}$ and $k+1^\text{th}$ nearest neighbor.  
Now, note that the NIBH output $\delta$ satisfies $\max_{i,j} |d_H(\vecx_i,\vecx_j) - d(\vecx_i,\vecx_j)| \leq \delta$\footnote{Here we assume $\lambda = 1$ without loss of generality.}. 
Observe that $\Delta_k = d(\vecx_0, \vecx_{k+1}) - d(\vecx_0, \vecx_k) \geq 2\delta$ is a sufficient condition for $d_H(\vecx_0,\vecx_{k+1}) \geq d_H(\vecx_0,\vecx_k)$, since 
\begin{align*}
& d_H(\vecx_0,\vecx_{k+1}) - d_H(\vecx_0,\vecx_k) \\
& \quad \geq d(\vecx_0, \vecx_{k+1}) - \delta - d(\vecx_0, \vecx_k) - \delta \\
& \quad \geq 2\delta - 2\delta = 0,
\end{align*}
by the triangular inequality. This leads to
\begin{align*}
& P(d_H(\vecx_0,\vecx_{k+1}) < d_H(\vecx_0,\vecx_k)) \\
& \quad = 1 -  P(d_H(\vecx_0,\vecx_{k+1}) \geq d_H(\vecx_0,\vecx_k)) \\ 
& \quad \leq 1 - P(\Delta_k \geq 2\delta) = P(\Delta_k < 2\delta).
\end{align*}
Therefore, combining all the above and \fref{lem:sg}, the probability that the $k$-nearest neighbor is not preserved is bounded by
\begin{align*}
P(E) & \leq k(Q-k) P(d_H(\vecx_0,\vecx_{k+1}) < d_H(\vecx_0,\vecx_k)) \\ & \leq k(Q-k)P(\Delta_k < 2\delta) \\
& = k(Q-k) P(\Delta_k - \mathbb{E}_x[\Delta_k] < -(\mathbb{E}_x[\Delta_k] - 2\delta) ) \\
& \leq k(Q-k) e^{-c(\mathbb{E}_x[\Delta_k] - 2\delta)^2} \\
& \leq kQ e^{-c(\mathbb{E}_x[\Delta_k] - 2\delta)^2}.
\end{align*}
Now, let $kQ e^{-c(\mathbb{E}_x[\Delta_k] - 2\delta)^2} \leq \epsilon$, we have that the requirement for the $k$-nearest neighbors to be exactly preserved with probability at least $1-\epsilon$ is 
\begin{align*}
\mathbb{E}_x[\Delta_k] \geq 2 \delta + \sqrt{\frac{1}{c}\log \frac{Qk}{\epsilon}}.
\end{align*}
\end{proof}

\paragraph{Remark} 
Note that our bound on the number of nearest neighbor preserved $k$ depends on the final outcome of the NIBH algorithm in the value of $\delta$.  
In order to relate our result to the number of binary hash functions $M$ required, we can make use of \citep[Thm.~1.10]{yanivplan}.
For a bounded set $K \subset \mathbb{R}^N$ with diameter $1$, let $M \geq C \delta^{-6} w(K)^2$ where $w(K) = \mathbb{E}_x[\sup_{\vecx \in K}\langle g,\vecx \rangle]$ denotes the Gaussian width of $K$ and some constant $C$.  
Then, \citep[Thm.~1.10]{yanivplan} states that, with high probability, $h(x)$ as defined in (1) with a random matrix $\bW$ whose entries are independently generated from $\mathcal{N}(0,1)$ is a $\delta_0$-isometric embedding. 
Therefore, if we initialize NIBH with such a random $\bW$ which is likely to be $\delta_0$-isometric, then empirically (see Figure~\ref{fig:convergance}), the NIBH algorithm will learn a better embedding that is $\delta$-isometric with $\delta < \delta_0$. Therefore, we have that the number of hash functions $M$ required for $k$-nearest neighbor preservation is at least $M \sim (\mathbb{E}_x[\Delta_k] - \sqrt{\log(kQ/\epsilon})^{-6}w(\bX)^2$ assuming that the training dataset $\bX$ is properly normalized.

\section*{Empirical convergence of the NIBH algorithm}
Figure~\ref{fig:convergance} shows the empirical loss and the actual distortion parameter $\delta$ as a function of the iteration count $\ell$ in the NIBH algorithm as applied on $4950$ secants (i.e., $Q=$ 100) from the $\it{MNIST}$ dataset. 
The behavior of the empirical loss function closely matches that of $\delta$ as they gradually converge.
The curve empirically confirms that minimizing the loss function in each iteration of NIBH (using the ADMM framework) directly penalizes the non-convex loss function (distortion parameter $\delta$).
After initializing the NIBH algorithm with random entries for $\bW$, the value of the distortion parameter $\delta$ significantly drops in the first few iterations, empirically conforming that NIBH  
learns an embedding with significantly lower distortion parameter after a few iterations.

\begin{figure}[h]
\centering
\includegraphics[width=0.4\textwidth]{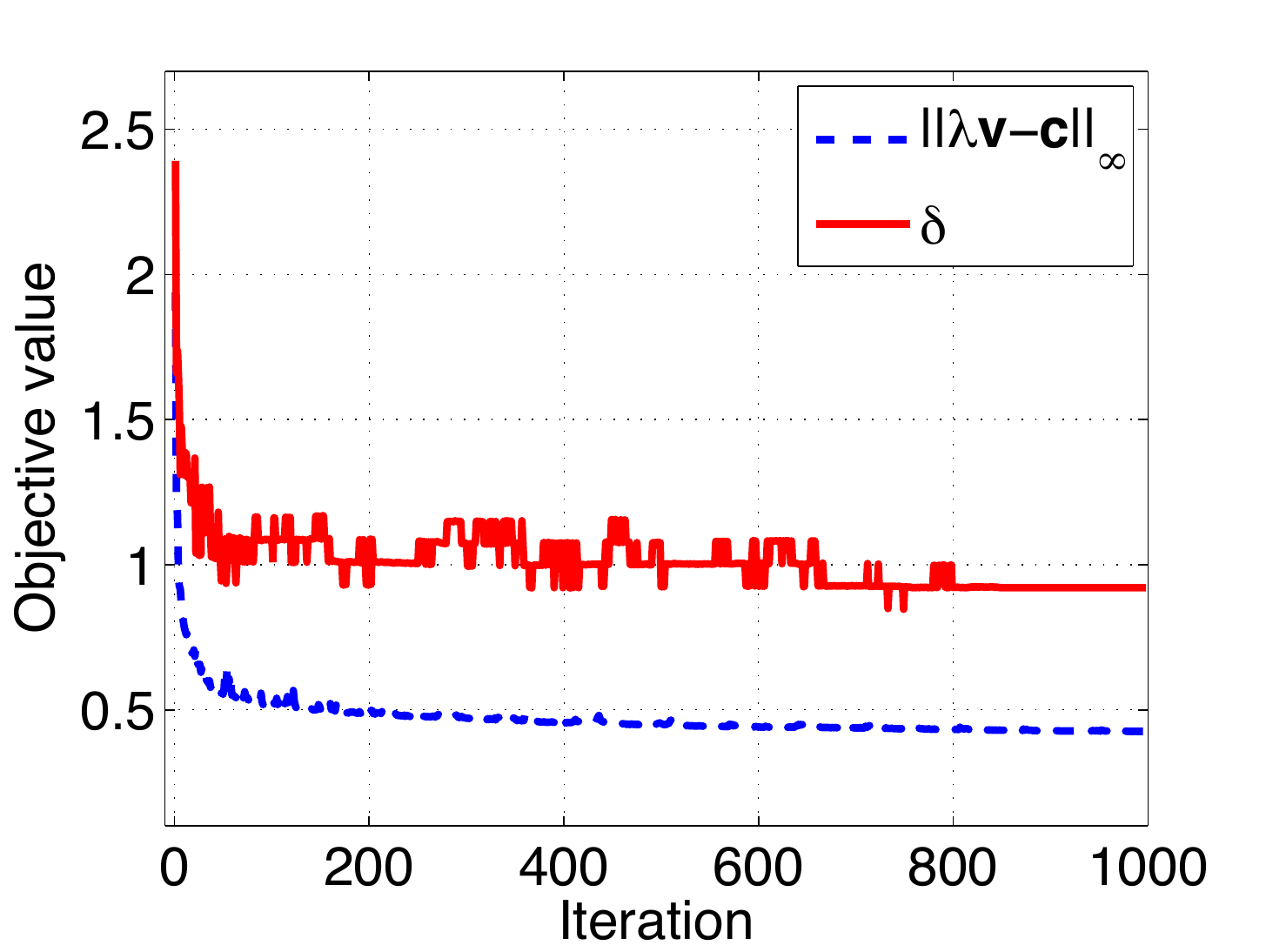}\label{fig:convergence}
\caption{Empirical convergence behavior of the NIBH algorithm.  Both the maximum distortion parameter $\delta$ and the loss function $||\lambda \vecv - \vecc||_{\infty}$ that approximates $\delta$ gradually decrease and converge as the number of iterations increases.  We see that the loss function of NIBH closely matches the behavior of the actual distortion parameter in each iteration.}%
\label{fig:convergance}
\end{figure}

\end{document}